\documentclass[aps,onecolumn]{IEEEtran}

\usepackage{amsmath,amsthm,amssymb}
\usepackage{enumerate}
\usepackage{color}
\usepackage{xcolor}
\usepackage{url}
\usepackage{balance}
\usepackage{graphicx} 
\usepackage{mathrsfs}
\usepackage{epsfig}
\usepackage{verbatim}
\usepackage{setspace}
\usepackage{cite}
\usepackage{bbm}

\usepackage{multicol}
\usepackage{lipsum}
\usepackage{etoolbox}
\usepackage{algpseudocode}
\usepackage{algorithmicx}
\usepackage{algorithm}
\usepackage{hyperref}

\usepackage{pgfplots}
  \pgfplotsset{compat=newest}
  \usetikzlibrary{plotmarks}
  \usetikzlibrary{arrows.meta}
  \usepgfplotslibrary{patchplots}
  \usepackage{grffile}
  \usepackage{amsmath}
\newtheorem{theorem}{Theorem}%[section]
\newtheorem{lemma}{Lemma}

\newtheorem{corollary}{Corollary}
\newtheorem{definition}{Definition}

\newtheorem{remark}{Remark}
\newtheorem{example}{Example}

\newcommand{\cX}{{\cal X}}
\newcommand{\cY}{{\cal Y}}
\newcommand{\cZ}{{\cal Z}}
\newcommand{\cS}{{\cal S}}

\newcommand{\removed}[1]{}

\makeatletter
\patchcmd{\@IEEEeqnarray}{\relax}{\relax\intertext@}{}{}
\makeatother

\algdef{SE}[DOWHILE]{Do}{doWhile}{\algorithmicdo}[1]{\algorithmicwhile\ #1}%

\title{Information Spectrum Converse for Minimum Entropy Couplings and Functional Representations}

\author{
\IEEEauthorblockN{
Yanina Y. Shkel$^{*}$
{{$^{\href{https://orcid.org/0000-0002-2575-1762}{\includegraphics[scale=0.08]{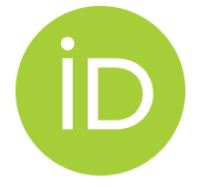}}}$}} $\&$ Anuj Kumar Yadav$^{*}$}{{$^{\href{https://orcid.org/0000-0001-7763-3787}{\includegraphics[scale=0.08]{orcid.JPG}}}$} }
\vspace{2mm}\\
\IEEEauthorblockA{
\{yanina.shkel, anuj.yadav\}@epfl.ch
\\
School of Computer \& Communication Sciences,\\
École Polytechnique Fédérale de Lausanne (EPFL)\\
Switzerland
}\thanks{\textbf{{($*$) : The paper follows alphabetical author order. Both the authors have equally contributed to this work.}}} 
}
\begin{document}
\maketitle 
\thispagestyle{empty}
\begin{abstract}
%THIS PAPER IS ELIGIBLE FOR THE STUDENT PAPER AWARD. 
Given two jointly distributed random variables $(X,Y)$, a functional representation of $X$ is a random variable $Z$ independent of $Y$, and a deterministic function $g(\cdot, \cdot)$ such that $X=g(Y,Z)$. The problem of finding a minimum entropy functional representation is known to be equivalent to the problem of finding a minimum entropy coupling where, given a collection of probability distributions $P_1, \dots, P_m$, the goal is to find a coupling $X_1, \dots, X_m$ ($X_i \sim P_i)$ with the smallest entropy $H_\alpha(X_1, \dots, X_m)$. This paper presents a new information spectrum converse, and applies it to obtain direct lower bounds on minimum entropy in both problems. The new results improve on all known lower bounds, including previous lower bounds based on the concept of majorization. In particular, the presented proofs leverage both - the information spectrum and the majorization - perspectives on minimum entropy couplings and functional representations.
\end{abstract}
%
%\begin{IEEEkeywords}

%\end{IEEEkeywords}
%
\newpage
\tableofcontents
\vspace{4mm}
\section{Introduction}\label{sec:introduction}
\subsection{Functional Representation Lemma}
\vspace{2mm}
The Functional Representation Lemma (FRL)~\cite{gamal_kim_2011} states that it is possible to decompose a random variable $X$ into an arbitrarily correlated random variable $Y$ and a residual random variable $Z$, independent of $Y$. More precisely, given any pair of random variables $(X,Y)\sim P_{XY}$, there exists a random variable $Z$, independent of $Y$ such that $X$ is a deterministic function of $(Y,Z)$. FRL has been widely used in proving various results in multi-user information theory~\cite{nit1,nit2,sfrl}, privacy and secrecy~\cite{yana},  and entropic causal inference~\cite{eci1,eci2}.

Several properties of functional representations have been studied in the literature. One such property of functional representations is the mutual information $I(X;Z)$. Bounds on $I(X;Z)$ have been studied in~\cite{sfrl} in the form of a strong functional representation lemma. This problem has also been studied in~\cite{radha} and~\cite{pub} with applications to one-shot channel simulation with unlimited common randomness. It has been applied to prove the result on constrained min-max remote prediction~\cite{minmax} and to prove the asymptotic achievability in Gelfand-Pinsker's Theorem~\cite{pinsker}. A closely related line of work is the setting of privacy funnel function~\cite{privacyf,privacyf2,privacyf3,privacyf4,privacyf5}. Here, the goal is to maximize $I(X;Z)$ while enforcing the Markov chain $Y\leftrightarrow X \leftrightarrow Z$. The motivation is to have perfect privacy for the information $Y$ while realizing maximum information about $X$, see~\cite{yana,privacyf4,privacyf5} for connections with FRL.

Another important property, which is the focus of this paper, is the entropy of the residual random variable $Z$. The minimization of $H(Z)$ has been applied in identifying the direction of causality~\cite{eci1,eci2}, and the compression length in a variable and fixed length secure compression~\cite{yana,yana2}. It could also be applied to studying a noisy communication channel~\cite{shannon}, if the goal is to minimize the entropy of the noise source. Lower entropy of $Z$ also means that less auxiliary randomness is needed to construct it from $(X,Y)$. This is of prime importance as randomness does not come for free and various expensive methods have been designed to generate it in practical systems~\cite{knuth,gm}. 
%In this paper, we focus on lower bounds on the R\'{e}nyi entropy $H_{\alpha}(Z)$, and prove a new converse in terms of the information spectrum of $Z$. This  converse improves upon the information spectrum converse in~\cite{yana} and the majorization converse in~\cite{cicalese} for R\'{e}nyi entropy of any order $\alpha$. 
%This work presents new lower bounds on the entropy of the residual random variable $Z$. In~\cite{yana}, Shkel et. al provided lower and upper bounds on $H(Z)$, inspired by the techniques from compression. In~\cite{cicalese}, Cicalese et.al provided a lower bound on $H(Z)$ based on the greatest lower bound w.r.t the majorization of probability distributions. 

\subsection{Minimum Entropy Coupling}
\vspace{2mm}
The minimum entropy coupling problem aims to find the joint distribution $(P_{X_1 X_2 \cdots X_m})$ with minimum joint R\'{e}nyi entropy $H_{\alpha}(X_1,X_2,\dots,X_m)$, given the $m$ marginal distributions $X_1\sim P_1, X_2 \sim P_2,\dots, X_m \sim P_m$ (see Definition~\ref{def:2}). This problem has been widely studied in~\cite{cicalese, senk,painsky,cheuk} and is closely related to the functional representation lemma. It has been shown in~\cite{eci2,cheuk} that finding the minimum entropy of $Z$ is the same as solving the minimum entropy coupling problem for the marginal distributions $\{P_{X|Y=y}\}_{y\in\cY}$ (see Lemma~\ref{lemma:1} and Appendix for details). Thus, the lower bounds on $H(Z)$ provided in this paper are also lower bounds on the joint entropy in the minimum entropy coupling problem. 

Additional scenarios where the problem of minimum entropy coupling arises include the entropic causal inference~\cite{eci2, eci1}, dimension reduction~\cite{vidyasagar}, contingency tables and transportation polytopes in~\cite{loera,dobra} and randomness generation~\cite{knuth,gm}. The computation of the minimum entropy coupling was shown to be NP-hard in~\cite{senk} and~\cite{vidyasagar}. In~\cite{cicalese}, it is shown that any coupling has the joint entropy at least $H(\land_{i=1}^m P_{i})$, where $\land$ denotes the greatest lower bound with respect to the majorization of probability distributions. A polynomial-time approximation algorithm for the construction of coupling within $\left \lceil{\log m}\right \rceil $ bits from  $H(\land_{i=1}^m P_{i})$ is also provided in~\cite{cicalese}. {{Li~\cite{cheuk} improved upon the result of~\cite{cicalese} by providing construction of a coupling with joint entropy being within $2-2^{2-m}$ bits from the  $H(\land_{i=1}^mP_{i})$. Recently, Compton~\cite{compton} showed that the greedy coupling algorithm provided in~\cite{eci1} is always within $\log_{2}(e)$ bits from $H(\land_{i=1}^m P_{i})$ which further improves upon the result in~\cite{cheuk} for $m>2$.}} 

{In this paper, we focus on lower bounds on the R\'{e}nyi entropy $H_{\alpha}(Z)$ (or equivalently $H_{\alpha}(X_1,X_2,\dots,X_m)$) for every $\alpha \in [0,\infty)$, and prove a new converse in terms of the information spectrum of $Z$. This  converse improves upon the information spectrum converse in~\cite{yana} and the majorization converse in~\cite{cicalese} for R\'{e}nyi entropy of any order $\alpha$. Similar results have also been recently studied in a parallel work by Compton et. al in~\cite{cmp}} (see Remark~\ref{rem:11} for details).\\

The rest of the paper is organized as follows. In Section~\ref{sec:system:model}, we state our notation, formulate the problem, and review the known lower bounds. We present our main results in Section~\ref{sec:main:results}. In Section~\ref{sec:comp}, we compare all the lower bounds and make concluding remarks. Finally, we postpone some of the proofs to Section~\ref{sec:proofs} and the Appendix.

\section{Preliminaries}\label{sec:system:model}
%In this section, we introduce some of the notation used throughout the paper. We formulate the problems of minimum entropy functional representations, and minimum entropy couplings. Finally, we review the known lower bounds for the two problems. 

\subsection{Notations}\label{subsec:notation}
We denote the probability mass function (PMF) of a random variable using a capital letter, say $P$, while the probability of an event is denoted using the bold-face letter $\mathbb{P}$ and the expectation of a random variable is denoted with $\mathbb{E}$. Given a random variable $X$, its support (and sets in general) is denoted by $\cX$, while a realization is denoted by lower case letter, for example, $x\in \cX$. The information of a random variable $X$ is
\begin{align}
   \imath_{X}(x):=\log\frac{1}{P_{X}(x)} \quad \forall x \in \cX,
\end{align}
where all logarithms have base two unless specified otherwise. The bold-face $\mathbb{F}_X$ denotes the cumulative distribution function (CDF) of $X$. With some abuse of notation, we use $\mathbb{F}_{\imath_X}$ to denote the CDF of the information $\imath_X(X)$ of $X$; $\mathbb{F}_{\imath_X}$ is also known as {\em the information spectrum} of $X$. Recall that the Shannon entropy of a random variable $X$ can be written as 
\begin{align}
    H(X)=\mathbb{E}\left[\imath_{X}(X)\right].
\end{align}
Similarly, the R\'{e}nyi entropy of $X$ can be written as
\begin{align}
H_{\alpha}(X)&=\frac{1}{1-\alpha}\log\left(\mathbb{E}\left[2^{(1-\alpha)\imath_{X}(X)}\right]\right)  \label{eq:ha}
\end{align}
for $\alpha \in [0,1)\cup(1,\infty)$ and $H_{\alpha}(X):=H(X)$ for $\alpha = 1$. Given a random variable $X$ with PMF $P_X$ we also use the notation $H(P_X)$ and $H_\alpha (P_X)$ to denote $H(X)$ and $H_\alpha (X)$. 
% The Renyi entropy for a discrete random variable $X$ is defined as
% \begin{align}
%    H_{\alpha}(X):=\frac{1}{1-\alpha}\log\Big(\sum_{x \in \cX}(P_{X}(x))^{\alpha}\Big) 
% \end{align}
% where the order $\alpha \in [0,1)\cup(1,\infty)$. It can be further re-written as
% \begin{align}
% H_{\alpha}(X)=&\frac{1}{1-\alpha}\log\left(\mathbb{E}\left[2^{(1-\alpha)\imath_{X}(X)}\right]\right)\label{eq:ha}
% \end{align}

%An important tool for analyzing ${H}^\star_{\alpha}(P_{XY})$ is the concept of majorization~\cite{major,cicalese,schur}. 
 Given two PMFs $P=(p_1,p_2, \cdots)$ and $Q=(q_1,q_2,\cdots)$ with countably infinite supports, let the probability masses for both be given in non-increasing order, that is $p_1 \geq p_2 \dots$ and $q_1 \geq q_2 \dots$. Then, we say that the distribution $P$ majorizes $Q$, or $Q \preceq_m P$, if 
\begin{align}
 \sum_{i=1}^k q_i \leq \sum_{i=1}^k p_i
\end{align}
for all $k \in \{1, 2, \dots\}$.
In the case that $P$ and $Q$ have finite supports with different cardinalities, we apply the same definition by padding both the PMFs with zeros.

\subsection{Problem Formulation}
Consider two jointly distributed random variables $(X,Y)\sim P_{XY}$ taking values in a countable set $\cX$ and a finite set $\cY$, respectively. Recall that the Functional Representation Lemma~\cite{gamal_kim_2011,yana,sfrl} says that there exists a random variable $Z$ taking values in some set $\cZ$, such that
    \begin{enumerate}
        \item $X$ is a deterministic function of $Y$ and $Z$ i.e.,
        \begin{align}
              H(X|Y,Z)=0 \label{eq:frl1}
        \end{align}
        \item $Y$ and $Z$ are independent of each other i.e.,
        \begin{align}
              I(Y;Z)=0.\label{eq:frl2}
        \end{align}
    \end{enumerate}
In other words, $X$ can be represented as a deterministic function of two independent random variables $Y$ and $Z$.
    
\begin{definition}[Minimum Entropy Functional Representation]\label{def:1}
The minimum R\'{e}nyi entropy of functional representation of $(X,Y) \sim P_{XY}$ is
\begin{align}
   {H}^\star_{\alpha}(P_{XY}):= \inf_{P_{Z|XY} \colon \substack{ H(X|Y,Z)=0, \\ I(Y;Z)=0}} H_{\alpha}(Z)
\end{align}
 for $\alpha \in [0,\infty)$.
\end{definition}
The coupling of $m$ PMFs $P_1, P_2, \dots, P_m$ refers to a joint distribution of $P_{X_1X_2 \cdots X_m}$ subject to the constraints $X_i\sim P_i$ for all $i \in \{1, 2, \dots, m\}$.
\begin{definition}[Minimum Entropy Coupling]\label{def:2}
The minimum R\'{e}nyi entropy of a coupling of a set of PMFs is
    \begin{align}
        H^\ast_{\alpha}(P_1,\dots&,P_m) \!:=\!
\inf_{P_{X_1X_2\cdots X_m} \colon {X_i} \sim P_{i}}H_{\alpha}(X_1,\dots,X_m)
    \end{align}
    for $\alpha \in [0,\infty)$.
\end{definition}

\begin{lemma}{\label{lemma:1}}
Let $(X,Y) \sim P_{XY}$ be jointly distributed random variables with $\cY = \{y_1, \dots, y_m\}$. Then
\begin{align}
     {H}^\star_{\alpha}(P_{XY}) =   H^\ast_{\alpha}(P_{X|Y =y_1},\dots,P_{X|Y = y_m}).
    \end{align}
That is, the problem of finding the functional representation with the minimum entropy is a minimum entropy coupling problem. 
\end{lemma}
\begin{proof}
The lemma is proved in~\cite{eci1,cheuk}. For the sake of completeness, the proof is presented in the Appendix.
\end{proof}
We state all results in terms of the functional representation conditions for the remainder of the paper. However, in light of Lemma~\ref{lemma:1}, we see that the same results also apply to the minimum entropy coupling problem.

\subsection{Prior Work}
It is shown in~\cite{cicalese} that
\begin{align}
           {H}^\star_{\alpha}(P_{XY}) \geq H_{\alpha}(\land_{y \in \cY}P_{X|Y=y}) \label{eq:cheuk}
\end{align}
where $\land$ represents the greatest lower bound with respect to majorization, $\preceq_m$, of the set $ \{P_{X|Y=y}\}_{y \in \cY}$. For details, refer to~\cite{major,cicalese}. Likewise,~\cite{yana,cicalese} independently presented the following lower bound for $\alpha =1$,
\begin{align}
       H^\star_{\alpha}(P_{XY}) \geq \sup_{y \in \cY} H_{\alpha}(P_{X|Y=y}).\label{eq:yana2}
\end{align}
Although the general case is not addressed in~\cite{yana,cicalese}, this lower bound could readily be extended to R\'{e}nyi Entropy of order $\alpha \in [0,\infty)$, see Section~\ref{sec:main:results} and Appendix for details.
%Note that  R\'{e}nyi Entropy is known to be Schur concave with respect to majorization: that is, if $Q \preceq_m P$ then $H_\alpha(Q) \geq H_\alpha(P)$. It follows from this that the left-hand-side of~(\ref{eq:cheuk}) is always greater than or equal to~ the left-hand-side of~(\ref{eq:yana2}). 

The following information spectrum converse is essentially shown in~\cite{yana}. Let $Z$ be any random variable satisfying~\eqref{eq:frl1} and~\eqref{eq:frl2}. Then, for any real valued $t \in [0,\infty)$ and $\tau > 0$,
\begin{align}
         \mathbb{P}[{\imath_{Z}(Z)\!>t}] \geq &\sup_{y \in \cY} \mathbb{P}[\imath_{X|Y}(X|Y)\!>t\!+\!\tau|Y\!=\!y]\!-\!\exp(-\tau). \label{eq:old_ic}
\end{align}
Our main contribution in this work is to strengthen~\eqref{eq:old_ic} by removing the dependence on $\tau$, as well as to leverage this new information spectrum converse in order to get direct lower bounds on ${H}^\star_{\alpha}(P_{XY})$.
\section{Main Results}\label{sec:main:results}
% To state our results below, we first define the following:
% $$G(t):=\max_{s \in \cS}\mathbb{P}[\imath_{X|S}(X|S)>t|S=s]$$
% $$G'(t):=\max_{s \in \cS}\mathbb{P}[\imath_{X|S=s}(X|s)\geq t]$$

% $$t':=2^{(1-\alpha)t}$$

%We state our main results in this section. We begin with our information spectrum converse in Theorem~\ref{thm:1}. We then introduce the idea of a lower bound with respect to the information spectrum, and connect it to majorization in Lemmas~\ref{lem:preceq} and~\ref{lem:greedy}. Then, we present two corollaries of our converse that give direct lower bounds on entropy of residual random variable in functional representations.

\subsection{Information Spectrum Converse}
%We begin with our information spectrum converse. 
\begin{theorem}\label{thm:1}
Let $(X,Y)$ be jointly distributed random variables supported on countable $\cal X$ and countable $\cal Y$. Consider any random variable $Z$ that satisfies~\eqref{eq:frl1} and~\eqref{eq:frl2}. Then, for any real-valued $t \in [0,\infty)$ we have,
\begin{align}
    \mathbb{P}[\imath_Z(Z) > t] \geq \sup_{y \in \cY}\mathbb{P}[\imath_{X|Y}(X|Y)>t|Y=y].
\end{align}
\end{theorem}

The proof of Theorem~\ref{thm:1} is given in Section~\ref{sec:proofs}. In this paper, we focus on countable $\cal X$ and finite $\cal Y$, though the results of Theorem~\ref{thm:1} could be extended to a more general setting. 

In order to connect Theorem~\ref{thm:1} to the majorization lower bound~\eqref{eq:cheuk}, we propose a similar definition with the ordering based on the information spectrum.
\begin{definition}[$\preceq_\imath$]
Consider two PMFs $P$ and $Q$ with countably infinite supports. We write $Q \preceq_\imath P$, if the information spectrum of $Q$ lower bounds the information spectrum of $P$. In other words, let $U \sim Q$ and $V\sim P$. Then
\begin{align}
 \mathbb{P} \left[\imath_U(U) \leq t \right] \leq \mathbb{P} \left[\imath_V(V) \leq t \right]
\end{align}
or equivalently,
\begin{align}
 \mathbb{P} \left[\imath_U(U) > t \right] \geq \mathbb{P} \left[\imath_V(V) > t \right]
\end{align}
for all $t\in [0, \infty)$.
\end{definition}

The following lemmas connect the above definition to the definition of majorization.
\begin{lemma}\label{lem:preceq}
Given two PMFs $Q$ and $P$ we have that 
\begin{align}
Q\preceq_\imath P \Rightarrow Q\preceq_m P
\end{align}
where $\preceq_m$ denotes majorization. 
\end{lemma}
The proof of Lemma~\ref{lem:preceq} is based on standard induction arguments, see Appendix for details. ({See also~\cite{li-idf}, where this result was independently proven.})

Note that in general, the reverse is not true: it is possible to have $Q\preceq_m P$ but not $ Q\preceq_\imath P$.

\begin{lemma}\label{lem:greedy}
Let $\{P_1, \dots, P_m\}$ be a finite collection of PMFs supported on finite $\{\cS_1, \dots, \cS_m\}$ and let
\begin{align}
   {\cal F} = \{Q \colon Q \preceq_\imath P_i \quad \forall i \in \{1, \dots, m\}\}.
\end{align}
There exists $Q^\ast \in {\cal F}$ such that $Q\preceq_m Q^\ast$ for all $Q \in {\cal F}$. In addition, $Q^\ast$ could be computed using a greedy algorithm with complexity linear in $\sum_{i=1}^m |\cS_i|$. This result could also be extended for a countable set $\{\cS_1, \cS_2,\dots\}$. We leave the complete details for future work. 
\end{lemma}
The idea behind the greedy construction of $Q^\ast$ is as follows. We construct $Q^\ast = (q^\ast_1, q^\ast_2, \dots)$ with $q^\ast_1 \geq q^\ast_2 \dots$ in a greedy way. First, assign 
\begin{align}
q^\ast_1 = \min_{i\in \{1, \dots, m\}} \left \{\max_{s\in \cS_i} P_i(s) \right\}.
\end{align}
This is the maximum value we can assign to $q^\ast_1$ while satisfying the information spectrum constraint defined by $\cal F$. We continue with this strategy: at each step $k$ we assign the maximum possible value to $q^\ast_k$ until the total probability of all $q^\ast_k$ reaches one. The fact that this distribution majorizes all other $Q\in {\cal F}$ is proven using a standard induction argument. Likewise, the number of steps the algorithm needs to take is linear in $\sum_{i=1}^m |\cS_i|$, since the constraints defined by $\cal F$ is a step function with at most steps $\sum_{i=1}^m |\cS_i|$, see the Appendix.

\subsection{Lower bounds on Shannon and R\'{e}nyi Entropy}
While Theorem~\ref{thm:1} does not provide direct lower bounds on $ H^\star_\alpha(P_{XY})$, the next two corollaries could be used to compute such lower bounds.
\begin{corollary}\label{cor:1}
Let $(X,Y)\sim P_{XY}$ be jointly distributed random variables supported on countable $\cal X$ and countable $\cal Y$. Then
\begin{align}
    H^\star_\alpha(P_{XY}) \geq K_\alpha(P_{XY}) \label{eq:cor_ren}
\end{align}
for any $\alpha \in [0, \infty)$, where
\begin{align}
K_{\alpha}(P_{XY})&:=\begin{cases}
\frac{1}{1-\alpha}\log\left[1+\int_{0}^{\infty}  J_{\alpha}(t)dt\right], &\text{if } \alpha\in [0,1)  \vspace{2mm}\\
& \text{or } \alpha \in (1,\infty) \\
\int_{0}^{\infty}G(t)dt ,  & \alpha=1 
\end{cases}\\
 \mbox{ with }G(t)&:=\sup_{y \in \cY}\mathbb{P}[\imath_{X|Y}(X|Y)>t|Y=y]\\
     \mbox{ and } J_{\alpha}(t)&:=\ln2 (1-\alpha) G(t)2^{(1-\alpha)t}.
\end{align}
 \end{corollary}
% \begin{corollary}\label{cor:1}
% For any jointly distributed $(X,Y)$ we have that 
% \begin{align}
%    H^\star_1(P_{XY})\geq \int_{0}^{\infty}G(t)dt
% \end{align}
% where $G(t):=\max_{y \in \cY}\mathbb{P}[\imath_{X|Y}(X|Y)>t|Y=y]$.
% % \begin{align}
% % G(t):=\max_{y \in \cY}\mathbb{P}[\imath_{X|Y}(X|Y)>t|Y=y]{\label{eq:gt}}
% % \end{align}
% Further, let the parameter $\alpha \in [0,1) \cup (1,\infty)$. Then, 
% \begin{align}
%     H^\star_\alpha(P_{XY}) \geq \frac{1}{1-\alpha}\log\left[1+\int_{0}^{\infty}  J_{\alpha}(t)dt\right]. \label{eq:cor_ren}
% \end{align}
% where $J_{\alpha}(t):=\ln2 (1-\alpha) G(t)2^{(1-\alpha)t}$.
% % \begin{align}
% %     J_{\alpha}(t):=\ln2 (1-\alpha) G(t)2^{(1-\alpha)t}. \label{eq:Jt}
% % \end{align}
% \end{corollary}
\begin{proof}
Theorem~\ref{thm:1} could be restated as
\begin{align}
    1-\mathbb{F}_{\imath_Z}(t) \geq G(t).\label{eq:asd}
\end{align}
Recall that for any non-negative random variable $V$, its expectation could be related to its CDF via
\begin{align}
    \mathbb{E}[V]=\int_{0}^{\infty}(1-\mathbb{F}_{V}(t))dt, \label{eq:excdf}
\end{align}
see, for example~\cite{hajek}. Since $\imath_Z(Z)\geq 0$, we have that
\begin{align}
    H(Z)&=\mathbb{E}[\imath_{Z}(Z)] =\int_{0}^{\infty}\Big(1-\mathbb{F}_{\imath_Z}(t)\Big) \geq \int_{0}^{\infty}G(t)dt,\label{eq:x2}
\end{align}
where~(\ref{eq:x2}) follows from~(\ref{eq:asd}). Equation~(\ref{eq:cor_ren}) is proved similarly for $\alpha \in [0,1) \cup (1, \infty)$, see Appendix for details.
\end{proof}

We remark that Corollary~\ref{cor:1} is different from all other lower bounds in this paper in that it does not involve computing the entropy of a random variable. For example, we compute the expectation of a random variable with CDF $1-G(t)$ in order to obtain a lower bound on Shannon entropy ($\alpha = 1$). In general, there may not be an information random variable with this CDF. The next corollary incorporates this additional constraint in its application of Theorem~\ref{thm:1}.

\begin{corollary}\label{cor:3}
Let $(X,Y)$ be jointly distributed random variables supported on finite $\cal X$ and finite $\cal Y$. Further, let the parameter $\alpha \in [0,\infty)$. Then, 
\begin{align}
   H^\star_{\alpha}(P_{XY}) \geq \inf_{Q \in {\cal F}} H_\alpha (Q) = H_\alpha(Q^\ast)
\end{align}
where ${\cal F} = \{Q \colon Q \preceq_\imath P_{X|Y=y} \forall y \in {\cal Y}\}$ and $Q^\ast \in {\cal F}$ is as given in Lemma~\ref{lem:greedy}.
\end{corollary}
\begin{proof}
The first inequality is a restatement of Theorem~\ref{thm:1} using the relation $\preceq_\imath$. According to Lemma~\ref{lem:greedy}, $Q^\ast$ majorizes every distribution in $\cal F$. R\'{e}nyi Entropy is Schur concave~\cite{schur} with respect to majorization: that is, if $Q \preceq_m P$ then $H_\alpha(Q) \geq H_\alpha(P)$. Thus, $Q^\ast$ achieves the infimum over $\cal F$.
\end{proof}
{{
\begin{remark}\label{rem:11}
Additionally, Note that for $\alpha =1$, these result have also been shown in a parallel work~\cite{cmp} using `profile method'. Our Corollary~\ref{cor:1} and Corollary~\ref{cor:3} are equivalent to Theorem $3.3$ and Theorem $5.5$ in~\cite{cmp}, respectively. The authors also showed that these results also hold  for a concave function $F:\mathbb{R}^d\rightarrow \mathbb{R}$, such that $F(p_1,p_2,p_3,...,p_n):=\sum_{i=1}^{n}f(p_i)$, where $f:\mathbb{R}\rightarrow\mathbb{R}$ is a non-negative and concave function satisfying $f(0)=0$.
 \end{remark}}}%

\section{Comparison of Lower Bounds}\label{sec:comp}
\pgfplotsset{plot coordinates/math parser=false}
\newlength\figureheight
\newlength\figurewidth

%In this section we compare the information spectrum lower bounds on entropy in Corollaries~\ref{cor:1} and~\ref{cor:3} to the lower bounds~(\ref{eq:cheuk}) and~(\ref{eq:yana2}). Corollary~\ref{cor:3} provably outperforms all previously known lower bounds, while Corollary~\ref{cor:1} outperforms previous bounds in many cases of interest.

\subsection{Analytical Comparison}

The following theorem compares the lower bounds in Corollary~\ref{cor:1} and~\ref{cor:3} to (\ref{eq:cheuk}) and~(\ref{eq:yana2}).% and clarifies the role of the relations $\preceq_m$ and $\preceq_\imath$. 
\begin{theorem}\label{thm:compare}
Let $(X,Y)$ be jointly distributed random variables supported on finite $\cal X$ and finite $\cal Y$. Let 
\begin{align}
{\cal F} = \{Q \colon Q \preceq_\imath P_{X|Y=y} \quad \forall y \in {\cal Y}\},
\end{align}
with $Q^\ast \in {\cal F}$ as in Lemma~\ref{lem:greedy}. Let $K_\alpha(P_{XY})$ be as in Corollary~\ref{cor:1}. We make the following statements for all $\alpha \in [0,\infty)$:
\begin{enumerate}
\item In general, we have that
\begin{align}
H_\alpha(Q^\ast) &\geq  K_\alpha(P_{XY}) \geq  \sup_{y \in \cY} H_\alpha(P_{X|Y=y} ) \label{eq:compare1c}
\end{align}
and,
\begin{align}
H_\alpha(Q^\ast) &\geq   H_{\alpha}(\land_{y \in \cY}P_{X|Y=y}) \label{eq:compare1d}\\
&\geq \sup_{y \in \cY}  H_\alpha(P_{X|Y=y}).\label{eq:compare1e}
\end{align}
\item  If there exists $y \in \cY$ such that $P_{X|Y=y} \preceq_\imath P_{X|Y= \hat y}$ for all ${\hat y} \in \cY$, then
\begin{align}
H_\alpha(Q^\ast)& = K_\alpha (P_{XY}) \label{eq:compare1aa}\\
&=  H_{\alpha}(\land_{y \in \cY}P_{X|Y=y}) \label{eq:compare1ab}\\
& =  \sup_{y \in \cY} H_\alpha(P_{X|Y=y} ).\label{eq:compare1ac} 
\end{align}
\item  If there exists $y \in \cY$ such that $P_{X|Y=y} \preceq_m P_{X|Y= \hat y}$ for all ${\hat y} \in \cY$, then
\begin{align}
H_{\alpha}(\land_{y \in \cY}P_{X|Y=y}) = \sup_{y \in \cY} H_\alpha(P_{X|Y=y}). \label{eq:compare1b}
\end{align}
In particular, in light of~\eqref{eq:compare1c}, this implies that
\begin{align}
H_\alpha(Q^\ast) &\geq  K_\alpha(P_{XY}) \geq H_{\alpha}(\land_{y \in \cY}P_{X|Y=y}).
\end{align}
\end{enumerate}
\end{theorem}
\begin{proof}
To see~\eqref{eq:compare1d}, note that for all $y\in \cY$ we have the following relation
\begin{align}
Q^\ast \preceq_m \land_{y \in \cY}P_{X|Y=y} \preceq_m P_{X|Y=y}.
\end{align}
The relations holds by Lemma~\ref{lem:preceq} together with $Q^\ast \in \cal F$, and the fact that $\land_{y \in \cY}P_{X|Y=y}$ is the greatest lower bound with respect to majorization. Therefore,~\eqref{eq:compare1d} holds by Schur concavity of R\'{e}nyi entropy. As already observed in~\cite{cicalese},~\eqref{eq:compare1e} and ~\eqref{eq:compare1b} hold by the same application of majorization and Schur concavity.
%As observed in~\cite{cicalese}, the second relation holds since $\land_{y \in \cY}P_{X|Y=y}$ is the greatest lower bound with respect to majorization of $\{P_{X|Y=y}\}_{y\in \cY}$. 
%Therefore,~\eqref{eq:compare1d} holds by Schur concavity of R\'{e}nyi entropy. If there exists $y \in \cY$ such that $P_{X|Y=y} \preceq_m P_{X|Y= \hat y}$ for all ${\hat y} \in \cY$, then again by by Schur concavity $P_{X|Y=y}$ will maximize R\'{e}nyi entropy for all $\alpha$ and this gives~\eqref{eq:compare1b}.

Let $U\sim Q^\ast$. We state the rest of the proof for the $\alpha = 1$ case. The other cases follow analogously.
For~\eqref{eq:compare1c}, we have that
\begin{align}
H(Q^\ast) & =  \int_{0}^{\infty} \mathbb{P}[\imath_U(U) > t] dt \label{eq:pos1}\\
&\geq \int_{0}^{\infty}\sup_{y \in \cY}\mathbb{P}[\imath_{X|Y}(X|Y)>t|Y=y]dt \label{eq:compare2aa}\\
&\geq \sup_{y \in \cY}\int_{0}^{\infty}\mathbb{P}[\imath_{X|Y}(X|Y)>t|Y=y]dt\label{eq:compare2ab} \\
&= \sup_{y \in \cY} H(P_{X|Y=y})\label{eq:pos2}
\end{align}
where~\eqref{eq:compare2aa} follows since $Q^\ast \in {\cal F}$, and~\eqref{eq:pos1} and~\eqref{eq:pos2} follow from~\eqref{eq:excdf}. If there exists $y \in \cY$ such that $P_{X|Y=y} \preceq_\imath P_{X|Y= \hat y}$ for all ${\hat y} \in \cY$ then $P_{X|Y=y}$ will have the same information spectrum as $Q^\ast$, and~\eqref{eq:compare2aa} and~\eqref{eq:compare2ab} become equalities. This shows that $H_\alpha(Q^\ast) = K_\alpha (P_{XY})= \sup_{y \in \cY} H_\alpha(P_{X|Y=y})$.
Finally,~\eqref{eq:compare1ab} holds by~\eqref{eq:compare1b} and Lemma~\ref{lem:preceq}.
\end{proof}

% We summarize these analytical results by highlighting that there are three distinct possibilities for how the bounds in Corollaries
% ~\ref{cor:1} and~\ref{cor:2}, as well as see~(\ref{eq:cheuk}) and~(\ref{eq:yana2}) could behave. 
% \begin{enumerate}
% \item First, observe that in light of Lemma~\ref{lem:preceq}, if~(\ref{eq:compare1}) and~(\ref{eq:compare2}) hold with equality, then~(\ref{eq:maj}) also holds with equality. That is, if the lower bound with respect to $\preceq_\imath$ is in the set $\{P_{X|S=s}\}_{s \in \cS}$, then all three bounds are the same.
% \item Secondly, observe that if~(\ref{eq:maj}) holds with equality, but ~\ref{eq:compare1}) and~(\ref{eq:compare2}) do not hold with equality, then the information spectrum converse outperforms the other lower bounds. This happens when the lower bound with respect to the majorization $\preceq_m$ is in the set $\{P_{X|S=s}\}_{s \in \cS}$, but it is not a lower bound with respect to the relation $\preceq_\imath$.
% \item Finally, when the lower bound with respect to the majorization $\preceq_m$ is not in the set $\{P_{X|S=s}\}_{s \in \cS}$, Corollaries
% ~\ref{cor:1} and~\ref{cor:2}, as well as~(\ref{eq:cheuk}), outperform ~(\ref{eq:yana2}). In this case, numerical simulations show that Corollaries
% ~\ref{cor:1} and~\ref{cor:2} perform better for low value of $\alpha$ while~(\ref{eq:cheuk}) performs better for high values of $\alpha$. These simulations are discussed next. 
% \end{enumerate}

\subsection{Numerical Comparison}
%Add to the tikz code to scale
%[scale=0.4, transform shape]

As we see in Theorem~\ref{thm:compare}, the lower bound in Corollary~\ref{cor:3} outperforms all other lower bounds. It is not clear how Corollary~\ref{cor:1} compares to~\eqref{eq:cheuk}, in general. Numerical evaluations show that Corollary~\ref{cor:1} outperforms~\eqref{eq:cheuk} for low values of $\alpha$, but is outperformed by~\eqref{eq:cheuk} for high values of $\alpha$. This is demonstrated by the following numerical example which has been carefully chosen to highlight the nuances between the different lower bounds.

\begin{figure}
\centering
\input{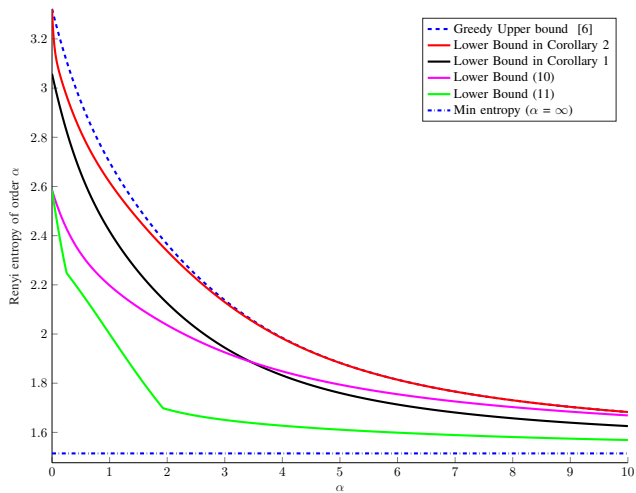}
\caption{Lower bound on R\'{e}nyi entropy for the collection of distributions in Example~\ref{example:1}. Note that Corollary~\ref{cor:3} outperforms all the other lower bounds, while Corollary~\ref{cor:1} outperforms~\eqref{eq:cheuk} for $\alpha \leq 3.445$.}
\label{fig:f1}
\end{figure}
\begin{example}[R\'{e}nyi entropy] \label{example:1}
Let $(X,Y)$ be jointly distributed random variables with $\cX = \{1, \dots, 6\}$, $\cY = \{y_1, y_2, y_3\}$ and the conditional PMF, 
\begin{align}
P_{X|Y}(\cdot|y_1) &= [0.5, 0.125, 0.125, 0.125, 0.125, 0]\\
P_{X|Y}(\cdot |y_2) &= [0.4, 0.4, 0.1, 0.1, 0, 0]\\
P_{X|Y}(\cdot|y_3) &= [0.35, 0.35, 0.25, 0.04, 0.005, 0.005].
\end{align}
Then,
\begin{align*}
&\land_{y \in \cY}P_{X|Y=y} \!=\! [0.35, 0.275, 0.125, 0.125, 0.12, 0.005],\\
&Q^\ast \! =\! [0.35, 0.15, 0.125, 0.125, 0.1, 0.1, 0.04, 0.005, 0.005].
\end{align*}
The comparison of lower bounds on R\'{e}nyi entropy for this example is provided in Figure~\ref{fig:f1}.
\end{example}

Observe that the lower bounds~\eqref{eq:cheuk} and~\eqref{eq:yana2} cannot exceed $\log |\cX|$ because they are limited to distribution with support size of $\cX$. Corollaries~\ref{cor:1} and~\ref{cor:3}, however, take into account possible support size enlargement for the functional representation variable $Z$. This can be seen in Figure~\ref{fig:f1} for $\alpha = 0$ (that is, log of the support size). This advantage of Corollaries~\ref{cor:1} and~\ref{cor:3} is even more apparent for binary supported distribution in the following example, see also Figure~\ref{fig:f2}.

\begin{figure}
\centering
% This file was created by matlab2tikz.
%
%The latest updates can be retrieved from
%  http://www.mathworks.com/matlabcentral/fileexchange/22022-matlab2tikz-matlab2tikz
%where you can also make suggestions and rate matlab2tikz.
%
\definecolor{mycolor1}{rgb}{1.00000,0.00000,1.00000}%
\begin{tikzpicture}[scale=0.6, transform shape]

\begin{axis}[%
width=6.028in,
height=4.754in,
at={(1.011in,0.642in)},
scale only axis,
xmin=0,
xmax=0.5,
xlabel style={font=\color{white!15!black}},
xlabel={p},
ymin=0.4,
ymax=1.4,
ylabel style={font=\color{white!15!black}},
ylabel={Shannon Entropy},
axis background/.style={fill=white},
legend style={at={(0.03,0.97)}, anchor=north west, legend cell align=left, align=left, draw=white!15!black}
]
\addplot [color=blue, dashed, very thick]
  table[row sep=crcr]{%
0	0.499915958164528\\
0.01	0.548260626725894\\
0.02	0.575160186084823\\
0.03	0.592904561192996\\
0.04	0.603938591670598\\
0.05	0.609259281426993\\
0.06	0.609259281426993\\
0.07	0.603938591670599\\
0.08	0.592904561192996\\
0.09	0.575160186084823\\
0.1	0.548260626725895\\
0.11	0.499915958164528\\
0.12	0.579018887323801\\
0.13	0.637957870336228\\
0.14	0.689182147659183\\
0.15	0.735336416007576\\
0.16	0.777675671846095\\
0.17	0.816938233663405\\
0.18	0.853611223393761\\
0.19	0.888040209388076\\
0.2	0.920482985684924\\
0.21	0.951139111175974\\
0.22	0.980167502961966\\
0.23	1.00769757530301\\
0.24	1.03383663814916\\
0.25	1.05867500476465\\
0.26	1.08228963181792\\
0.27	1.10474678374221\\
0.28	1.1261040276736\\
0.29	1.14641175661668\\
0.3	1.16571437230886\\
0.31	1.18405121757021\\
0.32	1.20145732089085\\
0.33	1.21796399801771\\
0.34	1.23359934305493\\
0.35	1.24838863308305\\
0.36	1.26235466428094\\
0.37	1.27551803320299\\
0.38	1.28789737370062\\
0.39	1.29950955763513\\
0.4	1.31036986577252\\
0.41	1.32049213391858\\
0.42	1.3298888783307\\
0.43	1.33857140365072\\
0.44	1.34654989598407\\
0.45	1.35383350326094\\
0.46	1.36043040462561\\
0.47	1.36634787028689\\
0.48	1.37159231300854\\
0.49	1.3761693322106\\
0.5	1.38008375148098\\
};
\addlegendentry{Greedy Upper bound ~\cite{eci1}}

\addplot [color=red, only marks, mark=asterisk, mark options={solid, red}, very thick]
  table[row sep=crcr]{%
0	0.499915958164533\\
0.03	0.592904561192997\\
0.06	0.609259281426994\\
0.09	0.575160186084823\\
0.12	0.579018887323801\\
0.15	0.735336416007576\\
0.18	0.853611223393761\\
0.21	0.951139111175974\\
0.24	1.03383663814916\\
0.27	1.10474678374221\\
0.3	1.16571437230886\\
0.33	1.21796399801771\\
0.36	1.26235466428094\\
0.39	1.29950955763513\\
0.42	1.3298888783307\\
0.45	1.35383350326094\\
0.48	1.37159231300854\\
};
\addlegendentry{Lower Bound in Corollary~\ref{cor:3}}

\addplot [color=black, dashdotted, ultra thick]
  table[row sep=crcr]{%
0	0.499915958164528\\
0.01	0.534510274350901\\
0.02	0.549104590537274\\
0.03	0.556150031702012\\
0.04	0.55829322291002\\
0.05	0.556791134352025\\
0.06	0.552384105239497\\
0.07	0.545561326925107\\
0.08	0.536670487655512\\
0.09	0.525971553712077\\
0.1	0.513666310539521\\
0.11	0.499915958164528\\
0.12	0.543169262316589\\
0.13	0.583949075973407\\
0.14	0.62251037501909\\
0.15	0.659060792183235\\
0.16	0.693772076590463\\
0.17	0.726788213231654\\
0.18	0.758231317836831\\
0.19	0.788206008312589\\
0.2	0.816802707274869\\
0.21	0.844100178386834\\
0.22	0.870167502961965\\
0.23	0.895065640662706\\
0.24	0.918848676413747\\
0.25	0.94156482728425\\
0.26	0.963257263438035\\
0.27	0.983964783396054\\
0.28	1.00372237393636\\
0.29	1.02256167775333\\
0.3	1.04051138669753\\
0.31	1.05759757447031\\
0.32	1.07384397967439\\
0.33	1.08927224785905\\
0.34	1.10390213946046\\
0.35	1.11775170918933\\
0.36	1.13083746136404\\
0.37	1.14317448485615\\
0.38	1.15477657065499\\
0.39	1.16565631452983\\
0.4	1.17582520684218\\
0.41	1.18529371121571\\
0.42	1.19407133348921\\
0.43	1.20216668214605\\
0.44	1.20958752122206\\
0.45	1.21634081653397\\
0.46	1.22243277593639\\
0.47	1.22786888420218\\
0.48	1.23265393302443\\
0.49	1.23679204655538\\
0.5	1.24028670282512\\
};
\addlegendentry{Lower Bound in Corollary~\ref{cor:1}}

\addplot [color=mycolor1, only marks, mark=+, mark options={solid, mycolor1}, very thick]
  table[row sep=crcr]{%
0	0.499915958164528\\
0.03	0.499915958164528\\
0.06	0.499915958164528\\
0.09	0.499915958164528\\
0.12	0.529360865287364\\
0.15	0.6098403047164\\
0.18	0.68007704572828\\
0.21	0.741482739931274\\
0.24	0.795040279384522\\
0.27	0.841464636208175\\
0.3	0.881290899230693\\
0.33	0.914926372779727\\
0.36	0.942683189255492\\
0.39	0.964799548505087\\
0.42	0.981453895033654\\
0.45	0.992774453987808\\
0.48	0.998845535995202\\
};
\addlegendentry{Lower Bound~(\ref{eq:cheuk})}

\addplot [color=green, very thick]
  table[row sep=crcr]{%
0	0.499915958164528\\
0.01	0.499915958164528\\
0.02	0.499915958164528\\
0.03	0.499915958164528\\
0.04	0.499915958164528\\
0.05	0.499915958164528\\
0.06	0.499915958164528\\
0.07	0.499915958164528\\
0.08	0.499915958164528\\
0.09	0.499915958164528\\
0.1	0.499915958164528\\
0.11	0.499915958164528\\
0.12	0.529360865287364\\
0.13	0.557438185027989\\
0.14	0.584238811642856\\
0.15	0.6098403047164\\
0.16	0.634309554640566\\
0.17	0.65770477874422\\
0.18	0.68007704572828\\
0.19	0.701471459883897\\
0.2	0.721928094887362\\
0.21	0.741482739931274\\
0.22	0.760167502961966\\
0.23	0.778011303546538\\
0.24	0.795040279384522\\
0.25	0.811278124459133\\
0.26	0.826746372492618\\
0.27	0.841464636208175\\
0.28	0.855450810560131\\
0.29	0.868721246339405\\
0.3	0.881290899230693\\
0.31	0.893173458377857\\
0.32	0.904381457724494\\
0.33	0.914926372779727\\
0.34	0.92481870497303\\
0.35	0.934068055375491\\
0.36	0.942683189255492\\
0.37	0.950672092687066\\
0.38	0.9580420222263\\
0.39	0.964799548505087\\
0.4	0.970950594454669\\
0.41	0.976500468757824\\
0.42	0.981453895033654\\
0.43	0.98581503717892\\
0.44	0.989587521222056\\
0.45	0.992774453987808\\
0.46	0.995378438820226\\
0.47	0.997401588567739\\
0.48	0.998845535995202\\
0.49	0.99971144175281\\
0.5	1\\
};
\addlegendentry{Lower Bound~(\ref{eq:yana2})}

\end{axis}
\end{tikzpicture}%
\caption{
Lower bound on Shannon entropy for collections of distributions (indexed by $p$) in Example~\ref{example:2}. Note that Corollary~\ref{cor:3} matches the upper bound and is therefore exactly $H_1^\star(P_{XY})$. Lower bounds~\eqref{eq:cheuk} and~\eqref{eq:yana2} match for this example. 
}
\label{fig:f2}
\end{figure}
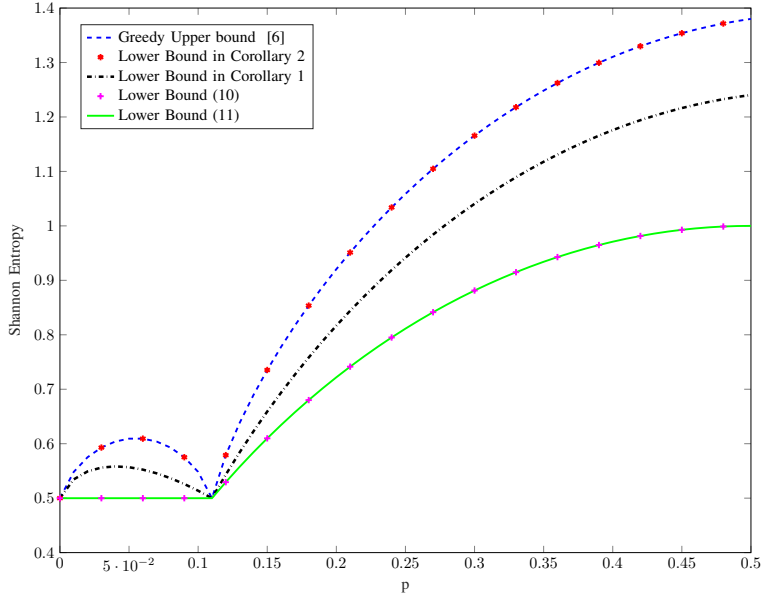

\begin{example}[Binary Support Size]\label{example:2}
Let $(X,Y)$ be jointly distributed random variables with $|\cX| = |\cY|=2$ and the conditional PMF, 
\begin{align}
P_{X|Y}(\cdot|y_1) &= [0.9, 0.1] \mbox{ and }P_{X|Y}(\cdot |y_2) = [1-p, p]
\end{align}
where $p\in[0,0.5]$. Then,
\begin{align}
\land_{y \in \cY}P_{X|Y=y} =\begin{cases}
[0.9, 0.1], &p\leq 0.1 \\
[1-p, p],  & p>0.1
\end{cases}
\end{align}
and 
\begin{align}
Q^\ast =\begin{cases}
[0.9, 0.1-p, p], &p\leq 0.1 \\
[1-p, p-0.1, 0.1],  & p>0.1.
\end{cases}
\end{align}
The comparison of lower bounds on Shannon entropy for this example is provided in Figure~\ref{fig:f2}.
\end{example}

\subsection{Concluding Remarks}
In order to provide context for the lower bounds in Figures~\ref{fig:f1} and~\ref{fig:f2} we include a plot of an upper bound. A number of theoretical upper bounds are available, see for example,~\cite {cheuk,cicalese,yana,eci1}. We implement the construction from~\cite{eci1} for our examples. We leave the question of improved theoretical upper bounds to future work and believe that combining information spectrum and majorization techniques could also be a fruitful strategy for this problem.
\section{Proof of Theorem~\ref{thm:1}}\label{sec:proofs}
%\subsection{Proofs}
\begin{proof}[Theorem~\ref{thm:1}]
Let's begin by fixing a $y \in \cY$ and a real-valued $t \in [0,\infty)$.
Now, let us define 
\begin{align}
\cX_y&:=\{x \in \cX: {P}_{X|Y}(x|y)>0\}\\
\mbox{ and }\cX_{y,t}&:=\{x \in \cX_y: {P}_{X|Y}(x|y)< 2^{-t} \} \\
&=\{x \in \cX_y: t<{\imath}_{X|Y}(x|y)<\infty \}.
\end{align}
Define also
\begin{align}
\cZ_{t}&:=\{z \in \cZ: {P}_{Z}(z)< 2^{-t} \} =\{z \in \cZ: {\imath}_{Z}(z)>t \}.   
\end{align}

Recall that $X$ is a deterministic function of $Y$ and $Z$ i.e., for some $g:\cY \times \cZ \rightarrow \cX$, we have $X=g(Y,Z)$. Further, define $\cZ_{y}(x):=\{z \in \cZ: g(y,z) = x\}$ and observe that
\begin{align}
P_{X|YZ}(x|y,z)=
\begin{cases}
1  &,\mbox{ if } z \in \cZ_{y}(x) \\
0    &,\mbox{ if } z \not \in \cZ_{y}(x) \label{eq:p1}
\end{cases}
\end{align}
i.e., for $y\in \cY$ fixed as above, every element of $\cX_y$ maps to a disjoint subset of $\cZ$. Therefore,
\begin{align}
     {P}_{X|Y}(x|y)&=\sum_{z \in \cZ}P_{X|YZ}(x|y,z){P}_{Z|Y}(z|y) \label{eq:p2}\\
 &=\sum_{z \in \cZ_{y}(x)}P_{X|YZ}(x|y,z){P}_{Z|Y}(z|y) \label{eq:p3}\\
  &\stackrel{(a)}{=}\sum_{z \in \cZ_{y}(x)}{P}_{Z|Y}(z|y)\stackrel{(b)}{=}\sum_{z \in \cZ_y(x)}{P}_{Z}(z).\label{eq:7a}
\end{align}
Eq.~\eqref{eq:p2} follows by the law of total probability. Eq.~\eqref{eq:p3} and ~\eqref{eq:7a}(a) follow from~\ref{eq:p1}, while~\eqref{eq:7a}(b) follows from the fact that $Y$ and $Z$ are independent. Putting everything together,
\begin{align}
\mathbb{P}&[\imath_{X|Y}(X|Y)>t|Y=y]
    =\sum_{x \in \cX_{y,t}}{P}_{X|Y}(x|y)\\
    &\stackrel{(a)}{=}\sum_{x \in \cX_{y,t}}\sum_{z \in \cZ_{y}(x)}{P}_{Z}(z) \stackrel{(b)}{=} \sum_{z \in \bigcup_{x \in \cX_{y,t}}\cZ_y(x)}{P}_{Z}(z)\label{eq:aa} \\
    &\stackrel{}{\leq} \sum_{z \in \cZ_{t}}{P}_{Z}(z) \stackrel{}{=}\mathbb{P}[\imath_{Z}(Z) > t].\label{eq:ac}
\end{align}
Eq.~(\ref{eq:aa})(a) follows from~(\ref{eq:7a}), Eq ~(\ref{eq:aa})(b) follows since every element of $\cX_y$ maps to a disjoint subset of $\cZ$, Eq.~(\ref{eq:ac}) follows from noting that $\forall$ $z \in \bigcup_{x \in \cX_{y,t}}\cZ_{y}(x)$,  we have $P_{Z}(z) < 2^{-t}$ and therefore $\bigcup_{x \in \cX_{y,t}}\cZ_y(x) \subseteq \cZ_{t}$.\\

Thus, we have shown that
\begin{align}
\mathbb{P}\left[\imath_{Z}(Z) > t\right] \geq    \mathbb{P}\left[\imath_{X|Y}(X|Y) > t|Y=y\right].  \label{eq:f}
\end{align}
Since (\ref{eq:f}) holds for $\forall$ $y \in \cY$, it holds for the supremum over $\cY$.\\
\end{proof}
\clearpage
\balance

\bibliographystyle{IEEEtran}
\bibliography{IEEEabrv,ref.tex.bib}

% Generated by IEEEtran.bst, version: 1.14 (2015/08/26)
\begin{thebibliography}{10}
\providecommand{\url}[1]{#1}
\csname url@samestyle\endcsname
\providecommand{\newblock}{\relax}
\providecommand{\bibinfo}[2]{#2}
\providecommand{\BIBentrySTDinterwordspacing}{\spaceskip=0pt\relax}
\providecommand{\BIBentryALTinterwordstretchfactor}{4}
\providecommand{\BIBentryALTinterwordspacing}{\spaceskip=\fontdimen2\font plus
\BIBentryALTinterwordstretchfactor\fontdimen3\font minus
  \fontdimen4\font\relax}
\providecommand{\BIBforeignlanguage}[2]{{%
\expandafter\ifx\csname l@#1\endcsname\relax
\typeout{** WARNING: IEEEtran.bst: No hyphenation pattern has been}%
\typeout{** loaded for the language `#1'. Using the pattern for}%
\typeout{** the default language instead.}%
\else
\language=\csname l@#1\endcsname
\fi
#2}}
\providecommand{\BIBdecl}{\relax}
\BIBdecl

\bibitem{gamal_kim_2011}
A.~El~Gamal and Y.-H. Kim, \emph{Network Information Theory}.\hskip 1em plus
  0.5em minus 0.4em\relax Cambridge University Press, 2011.

\bibitem{nit1}
B.~Hajek and M.~Pursley, ``Evaluation of an achievable rate region for the
  broadcast channel,'' \emph{IEEE Transactions on Information Theory}, vol.~25,
  no.~1, pp. 36--46, 1979.

\bibitem{nit2}
F.~Willems and E.~van~der Meulen, ``The discrete memoryless multiple-access
  channel with cribbing encoders,'' \emph{IEEE Transactions on Information
  Theory}, vol.~31, no.~3, pp. 313--327, 1985.

\bibitem{sfrl}
C.~T. Li and A.~E. Gamal, ``Strong functional representation lemma and
  applications to coding theorems,'' in \emph{2017 IEEE International Symposium
  on Information Theory (ISIT)}, 2017, pp. 589--593.

\bibitem{yana}
Y.~Y. Shkel, R.~S. Blum, and H.~V. Poor, ``Secrecy by design with applications
  to privacy and compression,'' \emph{IEEE Transactions on Information Theory},
  vol.~67, no.~2, pp. 824--843, 2021.

\bibitem{eci1}
\BIBentryALTinterwordspacing
M.~Kocaoglu, A.~Dimakis, S.~Vishwanath, and B.~Hassibi, ``Entropic causal
  inference,'' \emph{Proceedings of the AAAI Conference on Artificial
  Intelligence}, vol.~31, no.~1, Feb. 2017. [Online]. Available:
  \url{https://ojs.aaai.org/index.php/AAAI/article/view/10674}
\BIBentrySTDinterwordspacing

\bibitem{eci2}
\BIBentryALTinterwordspacing
S.~Compton, K.~Greenewald, D.~A. Katz, and M.~Kocaoglu, ``Entropic causal
  inference: Graph identifiability,'' in \emph{Proceedings of the 39th
  International Conference on Machine Learning}, ser. Proceedings of Machine
  Learning Research, K.~Chaudhuri, S.~Jegelka, L.~Song, C.~Szepesvari, G.~Niu,
  and S.~Sabato, Eds., vol. 162.\hskip 1em plus 0.5em minus 0.4em\relax PMLR,
  17--23 Jul 2022, pp. 4311--4343. [Online]. Available:
  \url{https://proceedings.mlr.press/v162/compton22a.html}
\BIBentrySTDinterwordspacing

\bibitem{radha}
\BIBentryALTinterwordspacing
P.~Harsha, R.~Jain, D.~McAllester, and J.~Radhakrishnan, ``The communication
  complexity of correlation,'' \emph{IEEE Trans. Inf. Theor.}, vol.~56, no.~1,
  p. 438–449, jan 2010. [Online]. Available:
  \url{https://doi.org/10.1109/TIT.2009.2034824}
\BIBentrySTDinterwordspacing

\bibitem{pub}
M.~Braverman and A.~Garg, ``Public vs private coin in bounded-round
  information,'' in \emph{Automata, Languages, and Programming}, J.~Esparza,
  P.~Fraigniaud, T.~Husfeldt, and E.~Koutsoupias, Eds.\hskip 1em plus 0.5em
  minus 0.4em\relax Berlin, Heidelberg: Springer Berlin Heidelberg, 2014, pp.
  502--513.

\bibitem{minmax}
C.~T. Li, X.~Wu, A.~Ozgur, and A.~El~Gamal, ``Minimax learning for remote
  prediction,'' in \emph{2018 IEEE International Symposium on Information
  Theory (ISIT)}, 2018, pp. 541--545.

\bibitem{pinsker}
S.~Gel'Fand, ``Coding for channels with random parameters,'' \emph{Probl.
  Contr. Inform. Theory}, vol.~9, no.~1, pp. 19--31, 1980.

\bibitem{privacyf}
A.~Makhdoumi, S.~Salamatian, N.~Fawaz, and M.~Médard, ``From the information
  bottleneck to the privacy funnel,'' in \emph{2014 IEEE Information Theory
  Workshop (ITW 2014)}, 2014, pp. 501--505.

\bibitem{privacyf2}
F.~d.~P.~{Calmon}, A.~{Makhdoumi}, M.~{M{\'e}dard}, M.~{Varia},
  M.~{Christiansen}, and K.~R. {Duffy}, ``Principal inertia components and
  applications,'' \emph{IEEE Transactions on Information Theory}, vol.~63,
  no.~8, pp. 5011--5038, Aug 2017.

\bibitem{privacyf3}
\BIBentryALTinterwordspacing
S.~Asoodeh, M.~Diaz, F.~Alajaji, and T.~Linder, ``Information extraction under
  privacy constraints,'' \emph{Information}, vol.~7, no.~1, 2016. [Online].
  Available: \url{https://www.mdpi.com/2078-2489/7/1/15}
\BIBentrySTDinterwordspacing

\bibitem{privacyf4}
\BIBentryALTinterwordspacing
B.~Rassouli and D.~Gündüz, ``Information-theoretic privacy-preserving schemes
  based on perfect privacy,'' 2023. [Online]. Available:
  \url{https://arxiv.org/abs/2301.11754}
\BIBentrySTDinterwordspacing

\bibitem{privacyf5}
\BIBentryALTinterwordspacing
A.~Zamani, T.~J. Oechtering, and M.~Skoglund, ``On the privacy-utility
  trade-off with and without direct access to the private data,'' 2022.
  [Online]. Available: \url{https://arxiv.org/abs/2212.12475}
\BIBentrySTDinterwordspacing

\bibitem{yana2}
Y.~Y. Shkel and H.~V. Poor, ``A compression perspective on secrecy measures,''
  \emph{IEEE Journal on Selected Areas in Information Theory}, vol.~2, no.~1,
  pp. 163--176, 2021.

\bibitem{shannon}
C.~E. Shannon, ``A mathematical theory of communication,'' \emph{The Bell
  System Technical Journal}, vol.~27, no.~3, pp. 379--423, 1948.

\bibitem{knuth}
D.~Knuth and A.~Yao, \emph{Algorithms and Complexity: New Directions and Recent
  Results}.\hskip 1em plus 0.5em minus 0.4em\relax Academic Press, 1976, ch.
  The complexity of nonuniform random number generation.

\bibitem{gm}
J.~Roche, ``Efficient generation of random variables from biased coins,'' in
  \emph{Proceedings. 1991 IEEE International Symposium on Information Theory},
  1991, pp. 169--169.

\bibitem{cicalese}
F.~Cicalese, L.~Gargano, and U.~Vaccaro, ``Minimum-entropy couplings and their
  applications,'' \emph{IEEE Transactions on Information Theory}, vol.~65,
  no.~6, pp. 3436--3451, 2019.

\bibitem{senk}
\BIBentryALTinterwordspacing
M.~Kova{\v c}evi{\'c}, I.~Stanojevi{\'c}, and V.~{\v S}enk, ``On the entropy of
  couplings,'' \emph{Information and Computation}, vol. 242, pp. 369--382,
  2015. [Online]. Available:
  \url{https://www.sciencedirect.com/science/article/pii/S0890540115000450}
\BIBentrySTDinterwordspacing

\bibitem{painsky}
A.~Painsky, S.~Rosset, and M.~Feder, ``Memoryless representation of markov
  processes,'' in \emph{2013 IEEE International Symposium on Information
  Theory}, 2013, pp. 2294--298.

\bibitem{cheuk}
C.~T. Li, ``Efficient approximate minimum entropy coupling of multiple
  probability distributions,'' \emph{IEEE Transactions on Information Theory},
  vol.~67, no.~8, pp. 5259--5268, 2021.

\bibitem{vidyasagar}
M.~Vidyasagar, ``A metric between probability distributions on finite sets of
  different cardinalities and applications to order reduction,'' \emph{IEEE
  Transactions on Automatic Control}, vol.~57, no.~10, pp. 2464--2477, 2012.

\bibitem{loera}
J.~A.~D. Loera and E.~D. Kim, ``Combinatorics and geometry of transportation
  polytopes: An update,'' \emph{arXiv: Combinatorics}, 2013.

\bibitem{dobra}
A.~Dobra and S.~E. Fienberg, ``Bounds for cell entries in contingency tables
  given marginal totals and decomposable graphs.'' \emph{Proceedings of the
  National Academy of Sciences of the United States of America}, vol. 97 22,
  pp. 11\,885--92, 2000.

\bibitem{compton}
S.~Compton, ``A tighter approximation guarantee for greedy minimum entropy
  coupling,'' in \emph{2022 IEEE International Symposium on Information Theory
  (ISIT)}, 2022, pp. 168--173.

\bibitem{cmp}
\BIBentryALTinterwordspacing
S.~Compton, D.~Katz, B.~Qi, K.~Greenewald, and M.~Kocaoglu, ``Minimum-entropy
  coupling approximation guarantees beyond the majorization barrier,'' in
  \emph{Proceedings of The 26th International Conference on Artificial
  Intelligence and Statistics}, ser. Proceedings of Machine Learning Research,
  F.~Ruiz, J.~Dy, and J.-W. van~de Meent, Eds., vol. 206.\hskip 1em plus 0.5em
  minus 0.4em\relax PMLR, 25--27 Apr 2023, pp. 10\,445--10\,469. [Online].
  Available: \url{https://proceedings.mlr.press/v206/compton23a.html}
\BIBentrySTDinterwordspacing

\bibitem{major}
A.~W. Marshall, I.~Olkin, and B.~C. Arnold, ``Inequalities: Theory of
  majorization and its applications,'' 1980.

\bibitem{li-idf}
C.~T. Li, ``Infinite divisibility of information,'' \emph{IEEE Transactions on
  Information Theory}, vol.~68, no.~7, pp. 4257--4271, 2022.

\bibitem{hajek}
B.~Hajek, \emph{Random Processes for Engineers}.\hskip 1em plus 0.5em minus
  0.4em\relax Cambridge University Press, 2015.

\bibitem{schur}
J.~E. Peari{'c}, F.~Proschan, and Y.~L. Tong, ``Convex functions, partial
  orderings, and statistical applications,'' 1992.

\end{thebibliography}
\newpage
\appendix

\subsection{Proof of Lemma~\ref{lemma:1}:}
\begin{proof}[Lemma~\ref{lemma:1}]
Let the set $\cY:=\{y_1,y_2,y_3,\cdots,y_{|\cY|}\}$. Recall that $X$ is a deterministic function of $Y$ and $Z$. Thus,
\begin{align*}
    X=g(Y,Z).
\end{align*}
For a fixed $y \in \cY$, we define the random variable $X_y$ as
\begin{align*}
    X_y &:= g(y,Z)\\
    &=g_{y}(Z)
\end{align*}
for some function $g_y: \cZ\rightarrow\cX$.\\
Further, note that the random variable $X_y$ is distributed according to $P_{X|Y=y}$. Now,
\begin{align}
    H(Z,g_{y_1}(Z),g_{y_2}(Z),\cdots,g_{y_{|\cY|}}(Z))\notag\\
    &\hspace*{-50mm}=H(g_{y_1}(Z),g_{y_2}(Z),\cdots,g_{y_{|\cY|}}(Z))\notag\\ &\hspace*{-30mm}+ H(Z|g_{y_1}(Z),g_{y_2}(Z),\cdots,g_{y_{|\cY|}}(Z))\label{eq:ent1}.
\end{align}
Also,
\begin{align}
    H(Z,g_{y_1}(Z),g_{y_2}(Z),\cdots,g_{y_{|\cY|}}(Z))\notag\\
    &\hspace*{-50mm}=H(Z)+ H(g_{y_1}(Z),g_{y_2}(Z),\cdots,g_{y_{|\cY|}}(Z)|Z)\label{eq:ent2}\\
    &\hspace*{-50mm}=H(Z)\label{eq:ent3}.
\end{align}
From~\eqref{eq:ent1} and~\eqref{eq:ent3}, we have
\begin{align}
    H(Z)&\geq H(g_{y_1}(Z),g_{y_2}(Z),\cdots,g_{y_{|\cY|}}(Z))\notag \\
    &=H(X_{y_{1}},X_{y_{2}},\cdots,X_{y_{|\cY|}})\label{ent:5}\\
    &\geq \inf_{P_{X_1X_2\cdots X_m} \colon {X_i} \sim P_{i}} H(X_{y_{1}},X_{y_{2}},\cdots,X_{y_{|\cY|}})\label{ent:6}.
\end{align}
Note that Eq.~\eqref{ent:5} holds for every random variable $Z$ satisfying~\eqref{eq:frl1} and~\eqref{eq:frl2}, and for every joint distribution of the random variables $(X_{y_{1}},X_{y_{2}},\cdots,X_{y_{|\cY|}})$. Therefore, the best possible lower bound on $H(Z)$ is obtained by minimizing the joint entropy  $H(X_{y_{1}},X_{y_{2}},\cdots,X_{y_{|\cY|}})$, which is equivalent to solving the minimum entropy coupling problem given the random variables $X_{y_{1}},X_{y_{2}},\cdots,X_{y_{|\cY|}}$ distributed according to $P_{X|Y=y_1}, P_{X|Y=y_2},\cdots,P_{X|Y=y_{|\cY|}}$, respectively.
\begin{remark}
    The above proof can be easily extended for the R\'{e}nyi entropy of order $\alpha \in [0,1)\cup(1,\infty)$. \\Let $Z,g_{y_1}(Z),g_{y_2}(Z),\cdots,g_{y_{|\cY|}}$ be random variables as above, then
     \begin{align}
        H_{\alpha}(Z,g_{y_1}(Z),g_{y_2}(Z),\cdots,g_{y_{|\cY|}}(Z))&= H_{\alpha}(Z)\notag\\&\geq H_{\alpha}(g_{y_1}(Z),g_{y_2}(Z),\cdots,g_{y_{|\cY|}}).\label{eq:lkl}
    \end{align}
{
Note that Equation~\ref{eq:lkl} holds from noting that for any two random variables $X$ and $Y$, the following holds $H_{\alpha}(X,Y) \geq H_{\alpha}(X)$ and $H_{\alpha}(X,Y) \geq H_{\alpha}(Y)$. Further, if $Y=f(X)$ then $H_{\alpha}(X,Y) = H_{\alpha}(X) \geq H_{\alpha}(Y)$.}\\
    Thus, we have the following  result from~\eqref{eq:lkl}:
    \begin{align}
       H_{\alpha}(Z) \geq \inf_{P_{X_1X_2\cdots X_m} \colon {X_i} \sim P_{i}} H_{\alpha}(X_{y_{1}},X_{y_{2}},\cdots,X_{y_{|\cY|}}).
    \end{align}
\end{remark}
\end{proof}
\subsection{Proof of Lemma~\ref{lem:preceq} and~\ref{lem:greedy}}
\begin{proof}[Lemma~\ref{lem:preceq}]
Define
\begin{align}
t_i = \log \frac{1}{q_i} \mbox{ and } s_i = \log \frac{1}{p_i},
\end{align}
and note that $t_1, t_2, \dots$ and $s_1, s_2, \dots$ are both increasing sequences. We prove the claim by induction. 

For the base case, suppose $p_1 < q_1$ and note that this implies that $s_1 > t_1$. Then
\begin{align}
 q_1 &\leq \mathbb{P} \left[\imath_U(U) \leq t_1 \right]\leq \mathbb{P} \left[\imath_V(V) \leq t_1 \right] = 0
\end{align}
which is a contradiction. Therefore $p_1 \geq q_1$.

For the inductive step, assume that for some $k$
\begin{align}
 \sum_{i=1}^k q_i \leq \sum_{i=1}^k p_i.
\end{align}
If $p_{k+1} \geq q_{k+1}$, then we are done. Else, for the case $p_{k+1} < q_{k+1}$, we note that $s_{k+1} > t_{k+1}$. Then
\begin{align}
 \sum_{i=1}^{k+1} q_i&\leq \mathbb{P} \left[\imath_U(U) \leq t_{k+1} \right] \leq \mathbb{P} \left[\imath_V(V) \leq t_{k+1} \right] \\
 &=  \sum_{i=1}^{k} 1\{s_i \leq t_k\}p_i < \sum_{i=1}^{k+1} p_i. 
\end{align}
This completes the proof.
\end{proof}

\begin{proof}[Lemma~\ref{lem:greedy}]
Define
\begin{align}
G(t):=\max_{i \in \{1, \dots, m\}}\mathbb{P}[\imath_{S_i}(S_i)>t]
\end{align}
where $S_i \sim P_i$. Note $G(t)$ is an increasing step function that goes from zero to one in at most $\sum_{i=1}^m |\cS_i|$ steps. Also note that every $Q\in {\cal F}$ satisfies 
\begin{align}
\sum_{i=1}^{k} q_i \leq G\left(\log \frac{1}{q_k}\right) \label{eq:gp1}
\end{align}
where, again, we assume $q_1 \geq q_2 \geq q_3 \dots$.

We construct $Q^\ast = (q^\ast_1, q^\ast_2, q^\ast_3, \dots)$ with $q^\ast_1 \geq q^\ast_2 \geq q^\ast_3 \dots$ in a greedy way. First, assign 
\begin{align}
q^\ast_1 = \min_{i\in \{1, \dots, m\}} \left \{\max_{s\in \cS_i} P_i(s) \right\}.
\end{align}
This is the maximum value we can assign to $q^\ast_1$ while satisfying the information spectrum constraint~\eqref{eq:gp1}. We continue with this strategy: at each step $k$ we assign the maximum possible value to $q^\ast_k$ until the total probability of all $q^\ast_k$ reaches one. We claim that the algorithm will terminate in at most $2\sum_{i=1}^m |\cS_i|$ iterations. Indeed, for every `step' in $G(t)$, the algorithm can assign at most two probability values to $Q^\ast$ before it needs to start assigning smaller values and move on to the next `step'.\\
%{\color{red} Y.S.: need to clean up this discussion a bit still.}\\

We also claim that the resulting distribution $Q^\ast$ will majorize all other distributions that satisfy the information spectrum constraint~\eqref{eq:gp1}. To this end, let $Q\in {\cal F}$ be arbitrary such distribution. We show that $Q \preceq_m Q^\ast$ inductively. For the base case, we have that $q^\ast_1 \geq q_1$ since by construction we assigned the maximum possible value to $q^\ast_1$.\\
For the inductive step, assume that for some integer $k$ we have that $\sum_{i=1}^k q_k \leq \sum_{i=1}^k q^\ast_k$. If $q^\ast_{k+1} \geq q_{k+1}$, we are done. Otherwise, define
\begin{align}
t_{k+1} = \log \frac{1}{q^\ast_{k+1}} \mbox{ and } s_{k+1} = \log \frac{1}{q_{k+1}},
\end{align}
and note that $s_{k+1} < t_{k+1}$ and $G(s_{k+1}) \leq G(t_{k+1})$. Let
\begin{align}
g_{k+1} = G(s_{k+1}) - \sum_{i=1}^{k} q^\ast_i
\end{align}
and note that $g_{k+1} < q_{k+1}$. This is a consequence of the greedy construction; if this was not the case, we would assign $q_{k+1}$ to $q^\ast_{k+1}$. Also, $q^\ast_{k+1} \geq g_{k+1}$, again by the property of the greedy construction. This is because we could always assign the value of $g_{k+1}$ to $q^\ast_{k+1}$ without violating the information spectrum constraint. Thus we have that
\begin{align}
\sum_{i=1}^{k} q^\ast_i + q^\ast_{k+1} \geq G(s_{k+1}) \geq \sum_{i=1}^{k+1} q_i
\end{align}
and this completes the proof.

\end{proof}

\subsection{Proof of Corollary~\ref{cor:1}:} 
\begin{proof}[Corollary~\ref{cor:1}]
\begin{remark}\label{rem:fff}
    Note that the slightly modified version of Equation (\ref{eq:f}) i.e.,
    \begin{align}
    \mathbb{P}\Big[\imath_{Z}(Z) \geq t\Big] \geq \sup_{y \in \cY} \mathbb{P}\Big[\imath_{X|Y}(X|Y)\geq t|Y=y\Big]\label{eq:fff}
\end{align}
is also true. The proof follows via the same construction except the fact that the definitions of the sets $\cX_{s,t}$ and $\cZ_t$ does not involves the strict inequality.
\end{remark}

From equation~\eqref{eq:ha}, we have
\begin{align}
H_{\alpha}(Z)=&\frac{1}{1-\alpha}\log\Big(\mathbb{E}[\imath^{\alpha}_Z(Z)]\Big)\label{eq:ha1}
\end{align}
where,
$${\imath^{\alpha}_{Z}}(Z):=2^{(1-\alpha)\imath_{Z}(Z)}$$

\underline{\textit{Case 1: }$\alpha \in [0,1)$}:\\

From Theorem (\ref{thm:1}), we have:
\begin{align}
    \mathbb{P}[\imath_Z(Z) > t] &\geq \sup_{y \in \cY}\mathbb{P}[\imath_{X|Y}(X|Y)>t|Y=y]\notag \\
    &=G(t) \notag
\end{align}
Note that $\mathbb{P}[2^{(1-\alpha)\imath_{Z}(Z)}>2^{(1-\alpha)t}]=\mathbb{P}[\imath_Z(Z) > t]$ holds for $0 \leq \alpha<1$. Therefore, 
    \begin{align}
    \mathbb{P}[\imath^{\alpha}_{Z}(Z) > t'] \geq G(t) \label{eq:wwe}
    \end{align}
where $t':=2^{(1-\alpha)t}$.\\
Now, 
\begin{align}
\mathbb{E}[\imath^{\alpha}_Z(Z)]&\stackrel{}{=}\int_{0}^{\infty}\Big(1-\mathbb{F}_{\imath^{\alpha}_Z}(t')\Big)dt'\label{eq:y1}\\
&\stackrel{}{=}\int_{0}^{\infty}\mathbb{P}[\imath^{\alpha}_{Z}(Z) > t']dt'\notag\\
&\stackrel{}{=}\int_{0}^{1}\mathbb{P}[\imath^{\alpha}_{Z}(Z) > t']dt'+\int_{1}^{\infty}\mathbb{P}[\imath^{\alpha}_{Z}(Z) > t']dt'\notag\\
&\stackrel{}{=}1+\int_{1}^{\infty}\mathbb{P}[\imath^{\alpha}_{Z}(Z) > t']dt'\label{eq:y2}\\
&\stackrel{}{\geq}1+\int_{1}^{\infty}G(t)dt'\label{eq:y3}\\
&\stackrel{}{=}1+(1-\alpha)\ln2\int_{0}^{\infty}G(t)2^{(1-\alpha)t}dt.\label{eq:hhh}
\end{align}
Here,
\begin{enumerate}
    \item[(\ref{eq:y1})] follows from noting that $\imath^{\alpha}_Z(Z)> 0$ and \eqref{eq:excdf}.
    \item[(\ref{eq:y2})] follows from noting that $\imath^{\alpha}_{Z}(Z) \in [1,\infty)$ for $\alpha \in [0,1)$.
    \item[(\ref{eq:y3})] follows from equation (\ref{eq:wwe}).
    \item[(\ref{eq:hhh})] follows from the change of variables in integration.
\end{enumerate}
On substituting equation (\ref{eq:hhh}) in equation (\ref{eq:ha1}), we have
\begin{align}
   H_{\alpha}(Z) &\geq \frac{1}{1-\alpha}\log\Bigg[1+(1-\alpha)\ln2\int_{0}^{\infty} G(t)2^{(1-\alpha)t} dt\Bigg] \label{eq:g1} \\
&=\frac{1}{1-\alpha}\log\Bigg[1+\int_{0}^{\infty}  J_{\alpha}(t)dt\Bigg].\notag
\end{align}
This completes the proof for the case of $ 0 \leq \alpha<1$.\\

\underline{\textit{Case 2: }$\alpha \in (1,\infty)$}:\\

Here, we will use the modified version (cf. Remark (\ref{rem:fff})) of Theorem (\ref{thm:1}) i.e.,
 \begin{align}
    \mathbb{P}\Big[\imath_{Z}(Z) \geq t\Big] \geq \sup_{y \in \cY} \mathbb{P}\Big[\imath_{X|Y}(X|Y)\geq t|Y=y\Big]:=G^{'}(t)\notag
\end{align}
Note that $\mathbb{P}[2^{(1-\alpha)\imath_{Z}(Z)}\leq 2^{(1-\alpha)t}]=\mathbb{P}[\imath_Z(Z) \leq t]$ holds for $\alpha>1$. Therefore,
 \begin{align}
    \mathbb{P}[\imath^{\alpha}_{Z}(Z) \leq t'] \geq G^{'}(t)\label{eq:ren}
    \end{align} 
where $t':=2^{(1-\alpha)t}$.\\
Now, 
\begin{align}
\mathbb{E}[\imath^{\alpha}_Z(Z)]&\stackrel{}{=}\int_{0}^{\infty}\Big(1-\mathbb{F}_{\imath^{\alpha}_Z}(t')\Big)dt'\label{eq:z1}\\
&\stackrel{}{=}\int_{0}^{\infty}\Big(1-\mathbb{P}[\imath^{\alpha}_{Z}(Z) \leq t']\Big)dt'\notag\\
&\stackrel{}{=}\int_{0}^{1}\Big(1-\mathbb{P}[\imath^{\alpha}_{Z}(Z) \leq t']\Big)dt' \notag \\ &\hspace{15mm}+\int_{1}^{\infty}\Big(1-\mathbb{P}[\imath^{\alpha}_{Z}(Z) \leq t']\Big)dt'\notag\\
&\stackrel{}{=}\int_{0}^{1}\Big(1-\mathbb{P}[\imath^{\alpha}_{Z}(Z) \leq t']\Big)dt'\label{eq:z2}\\
&\stackrel{}{\leq}\int_{0}^{1}\Big(1-G^{'}(t)\Big)dt'\label{eq:z3}\\
&\stackrel{}{=}(\alpha-1)\ln2\int_{0}^{\infty}(1-G^{'}(t))2^{(1-\alpha)t}dt \label{eq:z4}\\
&\stackrel{}{=}1+(1-\alpha)\ln2\int_{0}^{\infty}G^{'}(t)2^{(1-\alpha)t}dt \label{eq:hhh1}
\end{align}
Here,
\begin{enumerate}
    \item[(\ref{eq:z1})] follows from noting that $\imath_Z(Z)\geq 0$ and \eqref{eq:excdf}.
    \item[(\ref{eq:z2})] follows from noting that $\imath^{\alpha}_{Z}(Z) \in (0,1)$ for $\alpha \in (1,\infty)$.
    \item[(\ref{eq:z3})] follows from equation (\ref{eq:ren}).
    \item[(\ref{eq:z4})] follows from the change of variables of integration.
\end{enumerate}
On substituting equation (\ref{eq:hhh1}) in equation (\ref{eq:ha1}), we have
\begin{align}\label{eq:g2}
   H_{\alpha}(Z) \geq \frac{1}{1-\alpha}\log\Bigg[1+(1-\alpha)\ln2\int_{0}^{\infty}G^{'}(t)2^{(1-\alpha)t}dt \Bigg]
\end{align}
\begin{remark}
Note that the information spectrum CDF of a discrete random variable is discontinuous or right continuous (step function). Therefore, the computation of the integral in~\eqref{eq:g1} and~\eqref{eq:g2} yields the same value.
\end{remark}
Therefore, we have:
\begin{align}
   H_{\alpha}(Z) \geq \frac{1}{1-\alpha}\log\Bigg[1+\int_{0}^{\infty}J_{\alpha}(t)dt \Bigg].\notag
\end{align}
This completes the proof for the case of $\alpha>1$.\\
\end{proof}

\end{document}